\documentclass{llncs}
\usepackage{planargraph}
\usepackage{graphicx}
\usepackage{amsmath,amssymb}
\usepackage{algorithm}
\usepackage[noend]{algorithmic}

\renewcommand{\length}{\mbox{length}}

\newcommand{\inflow}{\textrm{inflow}}
\newcommand{\outflow}{\textrm{outflow}}
\newcommand{\diam}{\textrm{diameter}}

\usepackage{ifthen}
\newcommand{\iffull}{\ifthenelse{\equal{\fullversion}{true}}}
\newcommand{\ifnfull}{\ifthenelse{\equal{\fullversion}{false}}}

\newcommand{\fullversion}{true}
\newcommand{\arxiv}{true}

\begin{document}
\title{Multiple-Source Single-Sink Maximum Flow in Directed Planar
  Graphs in $O(\diam \cdot n \log n)$ Time}
\author{Philip N. Klein\thanks{Supported in part by NSF Grant
    CCF-0964037}\inst{} and Shay Mozes\thanks{Supported 
    by NSF Grant CCF-0964037 and by a Kanellakis fellowship}\inst{}}
\institute{Brown University, Providence, RI, USA,\\
\email{\{klein,shay\}@cs.brown.edu}}

\maketitle

\begin{abstract}
  We develop a new technique for computing maximum flow in directed planar
  graphs with multiple sources and a single sink that significantly
  deviates from previously known techniques for flow problems. 
This gives rise to  an $O(\diam \cdot n \log n)$ algorithm for the problem.
\end{abstract}

\section{Introduction}

The study of maximum flow in planar graphs has a long history.  In
1956, Ford and Fulkerson introduced the max $st$-flow problem, gave a
generic augmenting-path algorithm, and also gave a particular
augmenting-path algorithm for the case of a planar graph where $s$ and
$t$ are on the same face (that face is traditionally designated to be
the infinite face).  Researchers have since published many
algorithmic results proving running-time bounds on max $st$-flow for
(a) planar graphs where $s$ and $t$ are on the same face, (b)
undirected planar graphs where $s$ and $t$ are arbitrary, and (c)
directed planar graphs where $s$ and $t$ are arbitrary.  The best
bounds known are (a) $O(n)$~\cite{HKRS97}, (b)  $O(n \log \log
n)$~\cite{INSW11}, and (c) $O(n \log n)$~\cite{BorradaileK09},
where $n$ is the number of nodes in the graph.

This paper is concerned with the maximum flow problem in the presence
of multiple sources. In max-flow applied to general
graphs, multiple sources presents no problem: one can reduce the
problem to the single-source case by introducing an artificial source
and connecting it to all the sources.  However, as Miller and
Naor~\cite{MN95} pointed out, this reduction violates planarity unless
all the sources are on the same face to begin with.  Miller and Naor
raise the question of computing a maximum flow in a planar graph with
multiple sources and multiple sinks.  Until recently, 
the best known algorithm for computing multiple-source
max-flow in a planar graph is to use the reduction in conjunction with
a max-flow algorithm for general graphs.  That is, no
planarity-exploiting algorithm was known for the problem.
A few months after developing the technique described in this paper we
developed with collaborators an algorithm for the more general problem of maximum flow
in planar graphs with multiple sources and sinks~\cite{BKMNWN11} which
runs in $O(n \log^3 n)$ time and uses a recursive approach. Given these recent developments, the
algorithm presented here is mostly interesting for the new technique
developed.  

In this paper we present an alternative algorithm for the maximum flow
problem with multiple sources and a single sink in planar graphs that
runs in $O(\diam \cdot n \log n)$ time. The diameter of a graph is defined
as the maximum over all pairs of nodes of the minimum number of
edges in a path connecting the pair. Essentially, we start with a non-feasible flow that
dominates a maximum flow, and convert it into a feasible maximum
preflow by eliminating negative-length cycles in the
dual graph. The main algorithmic tool is a modification of an algorithm of
Klein~\cite{Klein05} for finding multiple-source shortest paths in
planar graphs with nonnegative lengths; our modification identifies and eliminates negative-length cycles.
This approach is significantly different than all previously known
maximum-flow algorithms. While the relation between flow in the primal graph
and shortest paths in the dual graph has been used in many algorithms
for planar flow, considering fundamentally non-feasible flows and
handling negative cycles is novel. We believe that this is an
interesting algorithmic technique and are hopeful it will be useful
beyond the current context.

\subsection{Applications}
Schrijver~\cite{Schrijver02} has written about the history of the
maximum-flow problem.
Ford and Fulkerson, who worked at RAND, were apparently motivated by
a classified memo of Harris and Ross on interdiction
of the Soviet railroad system. That memo, which was declassified, contains a diagram of a planar
network that models the Soviet railroad system and has multiple
sources and a single sink. 
%
%
%
 
A more realistic motivation comes from
selecting multiple nonoverlapping regions in a planar structure.
Consider, for example, the following image-segmentation problem.  A
grid is given in which each vertex represents a pixel, and edges
connect orthogonally adjacent pixels.  Each edge is assigned a cost
such that the edge between two similar pixels has higher cost than
that between two very different pixels.  In addition, each pixel is
assigned a weight.  High weight reflects a high likelihood that the
pixel belongs to the foreground; a low-weight pixel is more likely to
belong to the background.

The goal is to find a partition of the pixels into foreground and
background to minimize the sum 
\begin{eqnarray*}
\lefteqn{\mbox{weight of  {\it background} pixels}}\\
& + & \mbox{cost of
  edges between {\it foreground} pixels and {\em background} pixels}
\end{eqnarray*}
subject to the constraints that, for each component $K$ of foreground
pixels, the boundary of $K$ forms a closed curve in the planar dual
that surrounds all of $K$ (essentially that the component is simply connected).

This problem can be reduced to multiple-source, single-sink max-flow
in a planar graph (in fact, essentially the grid).  For each pixel
vertex $v$, a new vertex $v'$, designated a source, is introduced and
connected only to $v$.  Then the sink is connected to the pixels at
the outer boundary of the grid. See~\cite{BK04} for similar
applications in computer vision.

\subsection{Related Work}
Most of the algorithms for computing
maximum flow in general (i.e., non-planar) graphs
build a maximum flow by starting from the zero flow and iteratively
pushing flow without violating arc capacities. Traditional
augmenting path algorithms, as well as more modern blocking flow
algorithms, push flow from the source to the sink at
each iteration, thus maintaining a feasible flow (i.e., a flow
assignment that respects capacities and obeys conservation at non-terminals) at all times.
Push-relabel algorithms relax the conservation requirement and maintain a feasible
preflow rather than a feasible flow. However, none of these algorithms
maintains a flow assignment that violates arc capacities.

There are algorithms for maximum flow in planar graphs that do use
flow assignments that violate capacities~\cite{Reif83,MN95}. 
However, these violations are not fundamental in the sense that the
flow  does respect capacities up to
a circulation (a flow with no sources or sinks).  In other words, the
flow may over-saturate some arcs, but no cut is over-saturated.
We call such flows \emph{quasi-feasible} flows.

The value of a flow that does over-saturate some cuts is higher than
that of a maximum flow. This situation can be identified by detecting
a negative-length cycle in the dual of the residual graph.
One of the algorithms in~\cite{MN95} uses this property in a
parametric search for the value of the maximum flow. When a
quasi-feasible flow with maximum value is found, it is converted into
a feasible one. This approach is
not suitable for dealing with multiple sources because the size of the 
search space grows exponentially with the number of sources.

Our approach is the first to use and handle fundamentally infeasible
flows.
Instead of interpreting the existence of negative cycles
as a witness that a given flow should not be used to obtain a maximum
feasible flow, we use the negative
cycles to direct us in transforming a fundamentally non-feasible flow
into a maximum feasible flow. A negative-length cycle whose length
is $-c$ corresponds to a cut that is over saturated by $c$ units of
flow. This implies that the flow should be decreased by pushing
$c$ units of flow back from the sink across that cut.

\section{Preliminaries}

In this section we provide basic definitions and notions that are useful in presenting the algorithm. 
Additional definitions and known facts that are relevant to the proof of correctness and to the analysis are presented later on.

We assume the reader is familiar with the basic definitions of planar embedded graphs and their duals (cf.~\cite{BorradaileK09}).
Let $G= \langle V,A \rangle$ be a planar embedded graph with node-set $V$ and arc-set $A$. 
For notational simplicity, we assume here and henceforth that $G$ is connected and has
no parallel edges and no self-loops.
For each arc $a$ in the arc-set $A$, 
we define two oppositely directed darts, one in the same orientation as $a$ (which 
we sometimes identify with $a$) and one in the opposite orientation. 
We define $\rev(\cdot)$ to be the function that takes each dart to the corresponding dart in 
the opposite direction.
It is notationally convenient to equate the edges, arcs and darts of $G$ with the edges, arcs and darts of the dual $G^*$.
It is well-known that contracting an edge that is not a self-loop
corresponds to deleting it in the dual, and that a set of darts forms a simple directed cycle in $G$ iff it forms a 
simple directed cut in the dual $G^*$~\cite{Whitney1933}; see Fig.~\ref{fig:dual}.
\begin{figure}
\ifthenelse{\equal{\arxiv}{true}}{
\centerline{\includegraphics[scale=0.35]{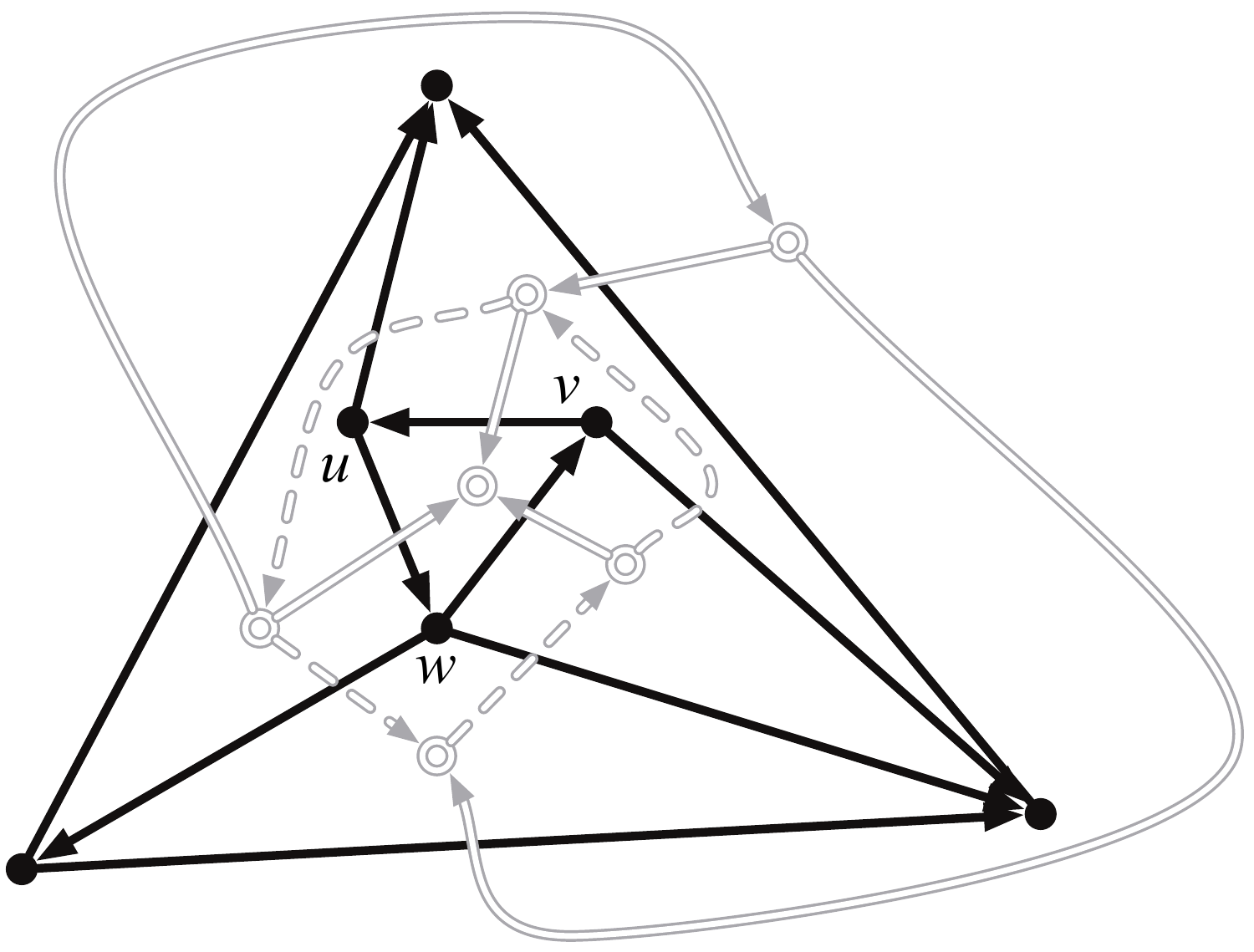}}
}{
\centerline{\includegraphics[scale=0.35]{figures/dualBW.pdf}}
}
\caption{A primal graph (black) and its dual (gray). The cut $\vec\delta(\{u,v,w\})$ in the primal 
corresponds to a counterclockwise dual cycle (dashed arcs).}
\label{fig:dual}
\end{figure}

Given a length assignment $\length(\cdot)$ on the darts, we extend it
to sets $D$ of darts by $\length(D) = \sum_{d \in D} \length(d)$.
For a set $X$ of nodes, let $\vec{\delta}(X)$ denote the set
of darts crossing the cut $(X, V \setminus X)$. Namely, 
$\vec{\delta}(X)  = \set{d : \tail(d) \in X, \head(d) \notin X}$. 
Let $T$ be a rooted spanning tree of $G^*$.  
For a node $v\in G^*$, let $T[v]$ denote the unique root-to-$v$ path in $T$.
The {\em reduced length of $d$ with
respect to $T$} is defined by 
\begin{equation}
\length_T(d) = \length(d) + \length(T[\tail_{G^*}(d)]) -
\length(T[\head_{G^*}(d)])
\end{equation}
The edges of $G$ not in $T$ form a spanning tree $\tau$ of $G$.
A dart $d$ is \emph{unrelaxed} if $\length_T(d) < 0$. Note
that, by definition, only darts not in $T$ can be unrelaxed. 
A \emph{leafmost unrelaxed} dart is an unrelaxed dart $d$ of $\tau$
such that no proper descendant of $d$ in $\tau$ is unrelaxed.
For a dart $d$ not in $T$, the \emph{elementary cycle of $d$ with
  respect to $T$ in $G^*$} is the cycle composed of $d$ and the unique path in
  $T$ between the endpoints of $d$.

\subsection{Flow}
Let $S \subset V$ be a set vertices called sources, and let $t\notin S$ be vertex called sink.

A {\em flow assignment} $f(\cdot)$ in $G$ is a real-valued function
on the darts of $G$ satisfying {\em antisymmetry}:
\begin{equation*} \label{eq:antisymmetry}
f(\rev(d))= -f(d)
\end{equation*}

A {\em capacity assignment} $c(\cdot)$ is a real-valued function on darts.
A flow assignment $f(\cdot)$ is {\em feasible} or {\em respects capacities}
if, for every dart $d$, $f(d) \leq c(d)$.  Note that, by antisymmetry,
$f(d) \leq c(d)$ implies $f(\rev(d)) \geq -c(d)$.  Thus a negative
capacity on a dart acts as a lower bound on the flow on the reverse
dart.  

For a given flow assignment $f(\cdot)$, the {\em net inflow} (or just
{\em inflow}) node $v$ is
$\inflow_f(v) = \sum_{d \in A: \head(d) = v} f(d)$.
\footnote{an equivalent definition, in terms of arcs, is 
$\inflow_f(v) = \sum_{a \in A : \head(a) =v} f(a) - \sum_{a \in A :
  \tail(a) = v} f(a) $.}
The {\em outflow} of $v$ is
$\outflow_f(v) = - \inflow_f(v)$.  The {\em value} of $f(\cdot)$ is the
  inflow at the sink, $\inflow_f(t)$.
A flow assignment $f(\cdot)$ is said to \emph{obey conservation} at
node $v$ if $\inflow_f(v) = 0$.
A flow assignment is a \emph{circulation} if it obeys conservation at
all nodes.
A flow assignment is a \emph{flow} if it obeys conservation at every node other than the sources and sink. 
It is a \emph{preflow} if for every node other than the sources, $\inflow_f(v) \geq 0$.   

For two flow assignments $f,f'$, the \emph{addition} $f+f'$ is the flow that
assigns $f(d)+f'(d)$ to every dart $d$.
A flow assignment $f$ is a \emph{quasi-feasible flow} if there exists
a circulation $\phi$ such that $f+\phi$ is a feasible flow. This
concept is not new, but it is so central to our algorithm that
we introduce a name for it.

The \emph{residual graph} of $G$ with respect to a flow assignment $f(\cdot)$ is the graph $G_f$ with
the same arc-set, node-set and sink, and with capacity assignment $c_f(\cdot)$
defined as follows. For every dart $d$, $c_f(d) = c(d) - f(d)$.

Given a feasible preflow $f_+$ in a planar graph, there exists an $O(n
\log n)$-time algorithm that converts $f_+$ into a feasible flow $f$
with the same value (cf.~\cite{JV82}). 
In fact, this can be done in linear time
by first canceling flow 
cycles using the technique of Kaplan and
Nussbaum~\cite{KaplanNussbaum2009}, and then
by sending any excess flow from back to
the sources in topological sort order.

\subsection{Quasi-Feasible Flows and Negative-Length Dual Cycles}
Miller and Naor prove that $f$ is a quasi-feasible flow in $G$ if and only if
$G^*_f$ contains no negative-length cycles. Intuitively, a
negative-length cycle in the dual corresponds to a primal cut whose
residual capacity is negative. That is, the corresponding cut is
over-saturated. Since a circulation obeys conservation at all nodes it
does not change the total flow across any cut. Therefore, there exists no
circulation whose addition to $f$ would make it feasible. 
Conversely, they show that if $G_f^*$ has no negative-length
cycles, then shortest path distances from any arbitrary node in
the dual define a feasible circulation in the primal.

\ifthenelse{\equal{\fullversion}{true}}{
\subsubsection {Potentials, Clockwise and Left-of}
The set of circulations in a graph forms a vector space, called the
{\em cycle space} of the graph.  For each face $f$, let $\uvec_f$ be the vector that
assigns flow value 1 to each dart in the clockwise boundary of the
face.  Each such vector is in the cycle space.  Moreover, the set
$\set{\uvec_f\ :\ f\neq f_\infty}$ is a basis for the cycle space.
Therefore every circulation can be written as a linear combination 
$\sum_f p_f \uvec_f$.  The coefficients $p_f$ are called {\em face
  potentials}.  By convention, the potential of $f_\infty$ is zero.

A circulation is {\em clockwise} ({\em counterclockwise}) if the
corresponding potentials are nonnegative (nonpositive).  (A
circulation can be neither.)  For $u$-to-$v$ paths $P$ and $Q$, we say
$P$ is {\em left of} ({\em right of})$Q$ if $P\rev(Q)$ is clockwise (counterclockwise).

\subsubsection{Price Functions and Reduced lengths}
For a directed graph $G$ with arc-lengths $\ell(\cdot)$,
a {\em price function} is a function $\phi$ from the nodes of $G$ to
the reals.  For an arc $uv$, the {\em reduced length with respect to
  $\phi$} is $\ell_\phi(uv) = \ell(uv) +\phi(u)-\phi(v)$.
For any
nodes $s$ and $t$, for any $s$-to-$t$ path $P$, $\ell_\phi(P) =
\ell(P)+\phi(s)-\phi(t)$.  This shows that an $s$-to-$t$ path is
shortest with respect to $\ell_\phi(\cdot)$ iff it is shortest with
respect to $\ell(\cdot)$. In particular, the reduced length of any
cycle is the same as its length.

\subsubsection{Winding Numbers}
For a curve $\gamma$ in the plane and a path $P$, the \emph{winding
  number} of $P$ about $\gamma$ is the number of times $P$ crosses
$\gamma$ right-to-left, minus the number of times it crosses $\gamma$
left-to-right.
}

\section{The Algorithm}

We describe an algorithm that, given a graph $G$ with $n$ nodes, a
sink $t$ incident to the infinite face $f_\infty$, and multiple sources, computes a maximum flow  
from the sources to $t$ in time $O(\diam \cdot n \log n)$, where $\diam$ is the
diameter of the face-vertex incidence graph of $G$.
Initially, each dart $d$ has a non-negative capacity, which we denote by
$\length(d)$ since we will interpret it as a length in the dual.
During the execution of the algorithm, the length assignment
$\length(\cdot)$ is modified. 
Even though the algorithm does not explicitly maintain a flow at all
times, we will refer throughout the
paper to the flow pushed by the algorithm. 
At any given point in the execution of the algorithm we can interpret the lengths of
darts in the dual as their residual capacities in the primal. By the flow pushed
by the algorithm, we mean the flow that would induce these residual capacities.

\begin{algorithm}[h!]\caption{Multiple-source single-sink maximum flow $(G,S,t,c_0)$}\label{alg:diameter}
{\bf Input:} planar directed graph $G$ with capacities $c_0$, source set $S$, sink $t$
incident to $f_\infty$\\
{\bf Output:} a maximum feasible flow $f$\\
\begin{algorithmic}[1]
\vspace{-10pt}
\STATE $\length(d) := c_0(d)$ for every dart $d$
\STATE initialize spanning tree $T$ of ${G^*}$ rooted at $f_\infty$ using right-first-search
\STATE let $\tau$ be the spanning tree of $G$ consisting of edges not
in $T$, and root $\tau$ at $t$
\FOR{each source $s\in S$} \label{initialize}
\FOR{each dart $d$ on the $s$-to-$t$ path in $\tau$} \label{initialize-iteration}
\STATE $\length(d) := \length(d)-\length(\delta(\set{s}))$
\STATE $\length(\rev(d)) := \length(\rev(d)) +
\length(\delta(\set{s}))$
\ENDFOR
\ENDFOR
\WHILE{there exist unrelaxed darts in ${G^*}$}
\STATE let $\hat d$ be an unrelaxed dart that is leafmost in $\tau$
\label{leafmost}
\IF [ //perform a pivot]{$\hat d$ is not a back-edge in $T$} \label{negcyc}
\STATE remove from $T$ the parent edge of $\head_{G^*}(\hat d)$ and insert
$\hat d$ into $T$ \label{relax} 
\ELSE [ // fix a negative cycle by pushing back flow]
\STATE let $C$ denote the elementary cycle of $\hat d$ with respect to
$T$ in $G^*$ \label{negcyc}
\STATE \algorithmicfor \ each dart $d$ of the $\hat d$-to-$t$ path of darts in the primal spanning tree $\tau$
\vspace{-5pt}
\begin{eqnarray*}
&& \length(d) :=  \length(d) + |\length(C)|\\
&& \length(\rev(d)) := \length(\rev(d)) - |\length(C)|
\end{eqnarray*}\label{fix} 
\vspace{-18pt}
\FOR{every dart $d$ strictly enclosed by $C$}
\STATE $f(d) := c_0(d) - \length_T(d)$ \label{record}
\STATE in $G$, contract $d$ \label{delete}
\COMMENT{//in ${G^*}$, delete $d$}
\ENDFOR
\ENDIF
\ENDWHILE
\STATE $f(d) := c_0(d) - \length_T(d)$ for every dart $d$ \label{final}
\STATE convert the preflow $f$ into a flow.\label{conv}
\end{algorithmic}
\end{algorithm}

The algorithm starts by pushing an infeasible flow that saturates
$\delta(s)$ from every source $s$ to $t$ (Line~\ref{initialize}).
It then starts to reduce that flow in order to make it feasible. 
This is done by using 
a spanning tree $\tau$ of $G$ and a
spanning tree $T$ of ${G^*}$ such that each edge is in exactly one of
these trees. 
The algorithm repeatedly identifies a negative
cycle $C$ in ${G^*}$, which corresponds in $G$ to an over-saturated cut.
Line~\ref{fix} decreases the lengths of
darts on a primal path in $\tau$ that starts at the sink $t$ and ends 
at some node $v$ (a face in ${G^*}$) that is enclosed by $C$. We call
such a path a \emph{pushback path}. This change corresponds to pushing
flow back from the sink to $v$ along the pushback path, making
the over-saturated cut exactly saturated.
We call the negative-turned-zero-length cycle $C$ a \emph{processed cycle}.
Processed cycles enclose no negative cycles. This implies
that there exists a feasible preflow  that saturates the cut
corresponding to a processed cycle (see Lemma~\ref{lem:flow-in-saturated-cycle}).
The algorithm records that preflow (Line~\ref{record}) and contracts the source-side of
the cut into a single node referred to as a \emph{super-source}.

When no negative length cycles are left in the contracted graph, the
flow pushed by the algorithm is quasi-feasible.
That is, the flow pushed by the algorithm is
equivalent, up to a circulation, to a feasible flow in the
contracted graph.
Combining this feasible flow in the contracted graph and the preflows recorded at
the times cycles were processed yields a maximum feasible preflow for
the original uncontracted graph.
In a final step, this feasible preflow is converted into a feasible
flow.

We now describe in more detail how negative cycles are identified and
processed by the
algorithm.  The algorithm maintains a spanning tree $T$ of ${G^*}$ rooted
at the infinite face $f_\infty$ of $G$, and a
spanning tree $\tau$ of $G$, rooted at the sink $t$. The tree $\tau$ consists of the edges not in $T$.
The algorithm tries to transform $T$ into a shortest-path tree by
pivoting into $T$ unrelaxed darts according to some particular order (line~\ref{leafmost}). 
However, if an unrelaxed dart $\hat d$ happens to be a
back-edge\footnote{a non-tree dart $d$ is a back edge of $T$ if
  $\head(d)$ is an ancestor in $T$ of $\tail(d)$.} in
$T$, the corresponding elementary cycle $C$ is a negative-length
cycle. To process $C$, flow is pushed from $t$ to $\head(\hat
d)$ along the $t$-to-$\hat d$ path in $\tau$ (line~\ref{fix}). The amount of flow
the algorithm pushes is $|\length(C)|$, so after $C$ is processed its
length is zero, and $\hat d$ is no longer unrelaxed.  
The algorithm then records a feasible preflow for $C$, deletes the
interior of $C$, and proceeds to find the next unrelaxed dart. See
Fig.~\ref{fig:pivot} for an illustration.
When all darts are
relaxed, $T$ is a shortest-path tree, which implies no more negative-length cycles
exist.
\begin{figure}
\ifthenelse{\equal{\arxiv}{true}}{
\centerline{\includegraphics[scale=0.3]{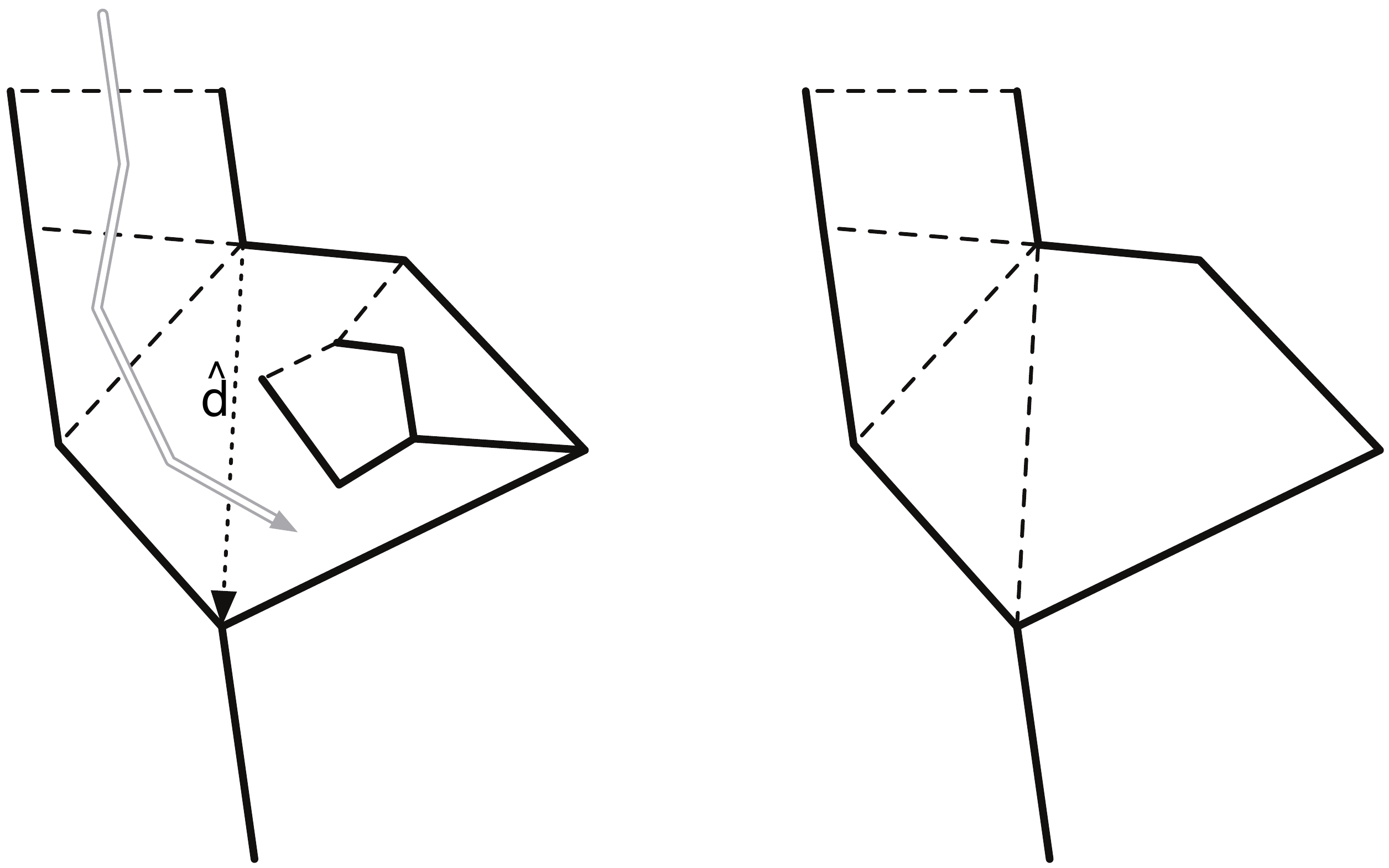}}
}{
\centerline{\includegraphics[scale=0.3]{figures/pivotBW.pdf}}
}
\caption{Identifying and processing a negative cycle. On the left, the tree $T$ is
  shown in solid. Non-tree edges are dashed.  The unrelaxed
  dart $\hat d$ is dotted. The elementary cycle of $\hat d$
  w.r.t. $T$ has negative length.
The primal pushback path is double-lined gray. After the negative
cycle is processed, its interior is deleted (on the right).}
\label{fig:pivot}
\end{figure}

To control the number of pivots we initialize $T$ to have a property
called \emph{right-shortness}, and show it is preserved by the
algorithm and that it implies that the number of pivots of any dart is bounded by
the diameter of the graph.
We later discuss how data structures enable the iterations to be
performed efficiently.

\section{Correctness and Analysis}\label{sec:analysis}
\ifnfull{
We refer the reader to the full version of this paper\footnote{available on the authors' website and at 
\texttt{http://arxiv.org/abs/1104.4728}} 
for complete
proofs of some of the lemmas in this section as well as for the precise definitions of
the following (standard) terms that are used in the sequel: clockwise
and counterclockwise, being left-of and right-of, reduced lengths, and
winding numbers.
}

We begin with a couple of basic lemmas.
Consider the sink $t$ as the infinite face of ${G^*}$. 
By our conventions, any clockwise (counterclockwise) cycle $C$ in ${G^*}$ corresponds to a primal
cut $(X, V \setminus X)$ such that $t \in X$ ($t \notin X)$; see Fig.~\ref{fig:dual}.
\begin{lemma}\label{lem:delta-cw}
Consider the sink $t$ as the infinite face of ${G^*}$. 
The length of any
clockwise (counterclockwise) dual cycle does not decrease (increase) when flow is pushed to $t$.
The length of any
clockwise (counterclockwise) dual cycle does not increase (decrease) when flow is pushed from $t$.
\end{lemma}

\iffull{
\begin{proof}
The change in the length of a dual cycle is equal to the change in the
residual capacity of the corresponding primal cut. Since clockwise cycles
correspond to cuts $(X,V \setminus X)$ s.t. $t \in X$, it follows that the
residual capacity of a cut corresponding to a clockwise dual cycle can
only increase when flow is pushed to $t$. The proofs of the other claims are similar.
\end{proof}
}

The following is a
restatement of a theorem of Miller and Naor~\cite{MN95}. 
\begin{lemma}\label{lem:flow-by-reduced-lengths}
Let $G$ be a planar graph. Let $c_0$ be a capacity function on
the darts of $G$. Let $f$ be a flow assignment. 
Define the length of a dart $d$ to be its residual capacity $c_0(d) - f(d)$.
If $f$ is quasi-feasible then $f'(d) = c_0(d)
- \length_T(d)$ is a feasible flow assignment, where $T$ is a shortest-path tree in $G^*$. 
Furthermore, $f'=f+\phi$ for some circulation $\phi$.
\end{lemma}
\iffull{
\begin{proof}
Miller and Naor~\cite{MN95} proved that there exist a feasible
circulation in $G$ if and only if ${G^*}$ contains no negative-length
cycles when residual capacities in $G$ are interpreted as lengths in
${G^*}$. Furthermore, they proved that the
flow assignment  $\phi(d) = \length(T[\head_{G^*}(d)]) - \length(T[\tail_{G^*}(d)])$ is
such a feasible 
circulation in $G$.

Therefore, if $f$ is quasi-feasible then $\phi(d)$ is a feasible
circulation with respect to the residual capacities $c(d) = c_0(d) - f(d)$.
Thus, $c(d) - \phi(d) \geq 0$. Interpret $c(d) - \phi(d)$ as
residual capacity with respect to some flow $f'$ in $G$. Namely,
$c(d) - \phi(d) = c_0(d) - f'(d)$. 
Therefore, $f'(d) =  c_0(d) -c(d) + \phi(d) = f(d) + \phi(d)$.
Furthermore, 
\begin{eqnarray*}
f'(d) &= & c_0(d) -c(d) + \phi(d) \\
& = & c_0(d) - c(d) + \length(T[\head_{G^*}(d)]) - \length(T[\tail_{G^*}(d)]) \\
& = & c_0(d) - \left(c(d) + \length(T[\tail_{G^*}(d)]) - \length(T[\head_{G^*}(d)]) \right)\\
& = & c_0(d) - \length_T(d) 
\end{eqnarray*}
where the last equality follows from the definition of the reduced
lengths with respect to $T$.
\end{proof}
}

\subsubsection{Correctness}
To prove the correctness of the algorithm, 
we first show that an elementary cycle w.r.t. a back-edge is indeed a
negative-length cycle.

\begin{lemma} Let $\hat d$ be an unrelaxed back-edge w.r.t. $T$. 
The length of the elementary cycle of $\hat d$ w.r.t. $T$ is negative.
\end{lemma} \label{lem:negcycle}
\begin{proof}
Consider the price function induced by from-root distances in
$T$.  Every tree  dart whose tail is closer to the
root than its head has zero reduced length. The reduced
length of an unrelaxed dart is negative. Since $\hat d$ is a
back edge w.r.t. $T$, its elementary cycle $C$ uses
darts of $T$ whose length is zero. Therefore the reduced length of $C$
equals the reduced length of just $\hat d$, which is negative.
The lemma follows since the length and reduced length of any cycle are
the same.
\end{proof}

\begin{lemma}\label{lem:fix}
Let $C$ be the negative-length cycle defined in
line~\ref{negcyc}. After $C$ is processed in line~\ref{fix},
$\length(C)=0$ and $\hat d$ is relaxed.
Furthermore, the following invariants hold just before line~\ref{fix}
is executed:
\begin{enumerate}
\item the flow pushed by the algorithm
satisfies flow conservation at every node other than the sources
(including super-sources) and the sink.
\item the outflow at the
sources is non-negative.
\item there are no clockwise negative-length cycles.
\end{enumerate}
\end{lemma}
\ifnfull{
The proof is omitted. The idea is that initially the invariants hold
since flow is pushed from the sources to the sink in the
initialization step. Subsequently, the invariants are maintained since
whenever a negative-length cycle is processed, flow is pushed
back into that cycle to make the corresponding cut exactly
saturated. Then the interior of the cycle is deleted.
}{ 
\begin{proof}
The proof is by induction on the number of iterations of the main while loop of the algorithm.
The initial flow pushed by the algorithm is from the sources to the
sink. Therefore, conservation is satisfied at every non-terminal, the
outflow at the sources is non-negative, and there is a one-to-one
correspondence between over-saturated cuts in the primal and 
counterclockwise negative-length cycles in the dual.

For the inductive step, since by the inductive assumption there are no
clockwise cycles, 
Lemma~\ref{lem:negcycle} implies that the cycle $C$ in line~\ref{negcyc} must be
counterclockwise.
Therefore, the unrelaxed dart $\hat d$ in line~\ref{leafmost} points in $\tau$
towards the root $t$. Thus, the length of $\hat d$ is increased in
line~\ref{fix} by $|\length(C)|$, and 
the length of $C$ becomes zero. This shows the main claim of the
lemma. 
To complete the proof of the invariants, note that the interior of $C$ is
deleted in ${G^*}$, so the flow pushed in line~\ref{fix} is now a flow from
the sink to the newly created super-source. This shows that invariant
(1) is preserved. Having $\length(C)=0$ implies that the total flow
pushed so far by the algorithm from the newly created super-source is
non-negative. This
shows (2). 

Since the outflow at all sources is 
non-negative after the pushback, the inflow at the sink must still be non-negative and equals
the sum of outflows at the sources. Therefore, for any negative-length
dual cycle, the corresponding primal
over-saturated cut $\Gamma(X)$
 must be such that $t \notin
X$. Hence any negative-length cycle is counterclockwise. This shows (3). 
\end{proof}
}

We now prove properties of the flow computed by the
algorithm. The following lemma characterizes the flow recorded in
line~\ref{record}. Intuitively, this shows that it
is a saturating feasible flow for the cut corresponding
to the processed cycle.

\begin{lemma}\label{lem:flow-in-saturated-cycle}
Let $C$ be a cycle currently being processed. 
Let $(X,V\setminus X)$ be the corresponding cut, where $t \notin X$. Let
$\vec \delta_c$ be the set
of darts crossing the cut.
The flow assignment $f$
computed in the loop in Line~\ref{record} satisfies:
\begin{enumerate}
\item $f(d)  \leq c_0(d)$ for all darts whose endpoints are both in $X$. \label{record-capacities}
\item every node in $X$ except sources and
  tails of darts of $\vec\delta_c$ satisfies conservation.\label{record-conservation}
\item for every $d' \in \vec\delta_c$, $\sum_{d : \head(d) = \tail(d')}
  f(d) \geq \sum_{d \in \delta_c : tail(d) = tail(d')} c_0(d)
  $ \label{record-saturate}
\end{enumerate}
\end{lemma}
\ifnfull{
The proof is omitted. The main idea is that at the time $C$
is processed it encloses no negative-length cycles. Furthermore, the tree $T$ is a shortest-path
tree for the interior of $C$ at that time. Since the length of the
cycle is adjusted to be zero, the corresponding cut is exactly
saturated, so by the relation of quasi-feasible flows, the distances
in $T$ define a feasible flow in which that cut is saturated.
}{
\begin{proof}
let ${G^*}_c$ denote the region of ${G^*}$
enclosed by $C$. 
Let $G_c$ denote the graph obtained from $G$ by contracting the sink
side of $\Gamma_c$ into a single node. Note that ${G^*}_c$ is the dual of
$G_c$.
Let $\hat{d}$ be the non-tree dart of $C$. 
Since $\hat{d}$ is leafmost unrelaxed, there are no unrelaxed darts in
${G^*}_c$ other than $\hat{d}$. By Lemma~\ref{lem:fix}, after the loop
in line~\ref{fix} is executed $\hat{d}$
is no longer unrelaxed as well. Therefore, the restriction of $T$ to ${G^*}_c$ is a shortest-path tree
for ${G^*}_c$, and ${G^*}_c$ contains no negative cycles. The conditions of
Lemma~\ref{lem:flow-by-reduced-lengths} are satisfied, so when $f(d)$
is set to 
$c_0(d) - \length_T(d)$ in Line~\ref{record}, it respects the capacities $c_0(d)$ for
all $d \in G_c$. This shows~\ref{record-capacities}.

By Lemma~\ref{lem:fix}, just before $C$ is processed, the
restriction of the flow pushed by the algorithm
to $G_c$ satisfies conservation everywhere except at the sources and
at the sink. In Line~\ref{fix} flow is pushed back to
$\tail(\hat{d})$, so there is more flow entering $\tail(\hat d)$ than
leaving it.
Therefore, the restriction of the flow pushed by the algorithm to
$G_c$ satisfies conservation everywhere except at the sources, the
sink (which in $G_c$ is the only node outside $\Gamma_c$) and
$\tail(\hat d)$. By Lemma~\ref{lem:flow-by-reduced-lengths}, the flow 
$f(d) = c_0(d) - \length_T(d)$ differs from the flow pushed by the algorithm
by a circulation, so $f(\cdot)$ satisfies conservation at all those
nodes as well. 

In particular, for every dart $d' \in \delta_c$, $0 \leq \sum_{d : \head(d) = \tail(d')}
  c_0(d) - \length_T(d)$ (this is an inequality only for $d' = \hat d$). 
However, note that the darts of $\delta_c$ are not strictly enclosed by
$C$, and that $\length_T(d)=0$ for every $d \in C$.
Since Line~\ref{record} assigns $f(d) =   c_0(d) - \length_T(d) $ only
to darts strictly enclosed by $C$ (and implicitly assigns zero to all
other darts), we get that, for the darts of
$\delta_c$, $ 0 \leq  \sum_{d : \head(d) = \tail(d')} f(d) + \sum_{d  \in \Gamma_c : \head(d) = \tail(d')} c_0(d)$.
Or alternatively, for every $d' \in \Gamma_c$, $\sum_{d : \head(d) = \tail(d')}
  f(d) \geq \sum_{d \in \Gamma_c : tail(d) = tail(d')} c_0(d) $. 
This shows~\ref{record-conservation} and ~\ref{record-saturate}.
\end{proof}
}

\begin{lemma}\label{lem:sources-in-saturatable-cycle}
The following invariant holds. In $G^*$ every source is enclosed
by some zero-length cycle that encloses no negative-length cycles.
\end{lemma}
\ifnfull{
The proof is omitted. The idea is that the initialization guarantees the
invariant holds initially. It is preserved since a processed cycle has
length zero and its interior is deleted and replaces with a super-source.
}{
\begin{proof}
Initially, the face of $G^*$ corresponding to each source is such a cycle.
At the time a cycle $C$ is processed and a super-source $s$ is created, 
$C$ is a zero-length cycle enclosing $s$ that encloses no
negative-length cycles.
The length of $C$ only changes if the pushback path ends at a dart of $C$ in the execution
of Line~\ref{fix}  when processing a cycle $C'$ at some later time. 
Since the interior of $C$ is deleted when $C$ is processed, 
$C'$ encloses $C$. At this time, the interior of $C'$, is contracted
into a single new super-source which is enclosed by the zero-length cycle
$C'$ that encloses no negative cycles.
\end{proof}
}
\iffull{
The following lemma is an easy consequence of
lemma~\ref{lem:sources-in-saturatable-cycle}.
}{
The following lemma, whose proof is omitted, is an easy consequence of
Lemma~\ref{lem:sources-in-saturatable-cycle}.
}
\begin{lemma}\label{lem:maximum}
The following invariant holds. There exists no feasible flow $f'$ s.t.
$\inflow_{f'}(t)$ is greater than $\inflow_f(t)$, where $f$ is
the flow pushed by the algorithm.
\end{lemma}
\iffull{
\begin{proof}
Let $f_s$ be the amount of flow source $s$ delivers to the sink in
the flow pushed by the algorithm. 
Consider a flow $f'$ and let $f'_s$ be the amount of flow source $s$
delivers to the sink in $f'$.
By Lemma~\ref{lem:sources-in-saturatable-cycle}, at any time along the
execution of the algorithm, every source is enclosed in a zero length
cycle. Consider such a zero-length cycle $C$.
If $\sum _{s \textrm{ enclosed by } C}  f'(s) > \sum _{s  \textrm{ enclosed by } C} f(s)$
 then the dual of the residual graph
w.r.t. $f'$ contains a negative cycle, so $f'$ is not
feasible. 
\end{proof}
}{}

\begin{lemma}\label{lem:maximum-feasible}
The flow $f$ computed in Line~\ref{final} is a maximum feasible flow
in the contracted graph $G$ w.r.t. the capacities $c_0$.
\end{lemma}
\begin{proof}
The flow pushed by the algorithm satisfies conservation by Lemma~\ref{lem:fix}.
By Lemma~\ref{lem:maximum}, there is no feasible flow of greater
value.
Since there are no unrelaxed darts, $T$ is a shortest-path tree in ${G^*}$
and ${G^*}$ contains no negative cycles. Therefore, the flow pushed by the
algorithm is quasi-feasible. This shows that the conditions
of Lemma~\ref{lem:flow-by-reduced-lengths} are satisfied. It follows
that $f$ computed in Line~\ref{final}  satisfies conservation, respects the capacities
$c_0$ and has maximum value. 
\end{proof}

\begin{lemma}\label{lem:uncontracted-preflow}
The flow assignment $f(d)$  is a
feasible maximum preflow in the (uncontracted) graph $G$ w.r.t.
the capacities $c_0$.
\end{lemma}
\begin{proof}
$f$ is well defined since each dart is assigned a
value exactly once; in Line~\ref{record} at the time it is contracted,
or in Line~\ref{final} if it was never contracted. By
Lemma~\ref{lem:maximum-feasible} and by part~\ref{record-capacities}
of Lemma~\ref{lem:flow-in-saturated-cycle}, $f(d) \leq c_0(d)$ for all
darts $d$, which shows feasibility.
By Lemma~\ref{lem:maximum-feasible} and by part~\ref{record-conservation} of
Lemma~\ref{lem:flow-in-saturated-cycle},  flow is conserved everywhere
except at the sources, the sink, and nodes that are tails of darts of
processed cycles. However, for a node $v$ that is the tail of some
dart of a processed cycle, $\sum_{d : head(d) = v}f(d) \geq 0$. This
is true by part \ref{record-saturate} of
Lemma~\ref{lem:flow-in-saturated-cycle}, and since $f(d) \leq c_0(d)$ for
any dart.  $f(\cdot)$ is therefore a preflow.
Finally, the value of $f(\cdot)$ is maximum by Lemma~\ref{lem:maximum-feasible}.   
\end{proof}

Lemma~\ref{lem:uncontracted-preflow} completes the proof of correctness
since line~\ref{conv} converts the maximum feasible
preflow into a feasible flow of the same value.

\subsubsection{Efficient implementation}

The dual tree $T$ is represented by a table $\text{parentD}[\cdot]$
that, for each nonroot node $v$, stores the dart $\text{parentD}[v]$
of $T$ whose head in $G^*$ is $v$.  The primal tree $\tau$ is
represented using a dynamic-tree data structure such as self-adjusting
top-trees~\cite{AHLT05}.  Each node and each edge of $\tau$ is
represented by a node in the top-tree.  Each node of the top-tree that
represents an edge $e$ of $\tau$ has two weights, $w_L(e)$ and
$w_R(e)$.  The values of these weights are the reduced lengths of the
two darts of $e$, the one oriented towards leaves and the one oriented
towards the root.\footnote{Note
that the length of the the cycle $C$ in line~\ref{negcyc} is exactly
the reduced length of $\hat d$.}
  The weights are represented so as to support
an operation that, given a node $x$ of the top-tree and an amount
$\Delta$, adds $\Delta$ to $w_R(e)$ and subtracts $\Delta$ from
$w_L(e)$ for all edges $e$ in the $x$-to-root path in $\tau$.

This representation allows each of the following
operations be implemented in $O(\log n)$ time:
lines~\ref{leafmost},~\ref{negcyc},~\ref{relax},~\ref{record}, and~\ref{delete}, the loop of
line~\ref{initialize-iteration}, and the loop of
line~\ref{fix}.\footnote{The whole initialization in the loop of
  line~\ref{initialize} can instead be carried out in linear time by
  working up from the leaves towards the root of $\tau$.}

\subsubsection{Running-Time Analysis}

Lines~\ref{record} and~\ref{delete} are executed at most once per edge.
To analyze the running time it therefore suffices to bound the number
of pivots in line~\ref{relax} and the number of negative cycles
encountered by the algorithm. Note that every negative length cycle
strictly encloses at least one edge. This is because the length of any
cycle that encloses just a single source is initially set to
zero, and since the length of the cycle that encloses just a single
super-source is set to zero when the corresponding negative cycle is
processed and contracted. 

Since the edges strictly enclosed by a processed cycle are deleted,
the number of processed cycles is bounded by the number of edges,
which is $O(n)$. 

It remains to bound the number of pivots.
We will prove that at the tree $T$ satisfies a
property called \emph{right-shortness}.
This property implies that the number of times a given dart pivots
into $T$ is bounded by the diameter of the graph.

\begin{definition}\cite{Klein05}
A tree $T$ is \emph{right-short} if for all nodes $v \in T$ there is
no \emph{simple} root-to-$v$ path $P$ that is:
as short as $T[v]$ and strictly right of $T[v]$
\end{definition}

Since $T$ is initialized using right first search, initially, for
every node $v$ there is no simple path in ${G^*}$ that is strictly
right of $T[v]$. Therefore, initially, $T$ is right-short. 
The algorithm changes $T$ in two ways; either by making a pivot
(line~\ref{relax}) or by processing a negative-length cycle
(line~\ref{fix}). 
\iffull{
The following lemma was proved in~\cite{Klein05}:
}

\begin{lemma}\cite{Klein05}\label{rs2}  leafmost dart relaxation
  (line~\ref{relax}) preserves right-shortness.
\end{lemma}

\begin{lemma}\label{rs1}  Right-shortness is preserved when processing
the counterclockwise negative-length  cycle  $C$ in line~\ref{fix}.
\end{lemma}
\ifnfull{
The proof is omitted. The main idea is that the changes in lengths of darts in line~\ref{fix}
correspond to pushing back flow from the sink $t$. 
By Lemma~\ref{lem:delta-cw}, the length of any counterclockwise
cycle does not decrease. Since lengths of darts of $T$ are not
affected by line~\ref{fix}, this implies that right-shortness is preserved.
}{
\begin{proof}
For a node $v$, let $P$ be denote the tree path $T[v]$. Let $Q$ be a simple path that is strictly right of $P$. By right-shortness of $T$, before line~\ref{fix} is executed $\length(Q) > \length(P)$. We show that this is also the case after line~\ref{fix} is executed.
By definition of right-of, $P \circ \rev(Q)$ is a clockwise
cycle. Note that since $C$ is counterclockwise, the changes in lengths of darts in line~\ref{fix}
correspond to pushing back flow from the sink $t$. 
Therefore, by Lemma~\ref{lem:delta-cw}, the length of $P \circ
\rev(Q)$ can only decrease. 
Since no dart of $P$ changes its length (only darts of $\tau$ are
changed in line~\ref{fix}), this implies that $rev(Q)$ can only
decreases in length. By antisymmetry, the length of $Q$ can only
increase, so $\length(Q) > \length(P)$ after the execution as well.
\end{proof}
}
We have thus established that
\begin{corollary}
$T$ is right-short at all times.
\end{corollary}

\begin{figure}
\ifthenelse{\equal{\arxiv}{true}}{
\centerline{\includegraphics[scale=0.3]{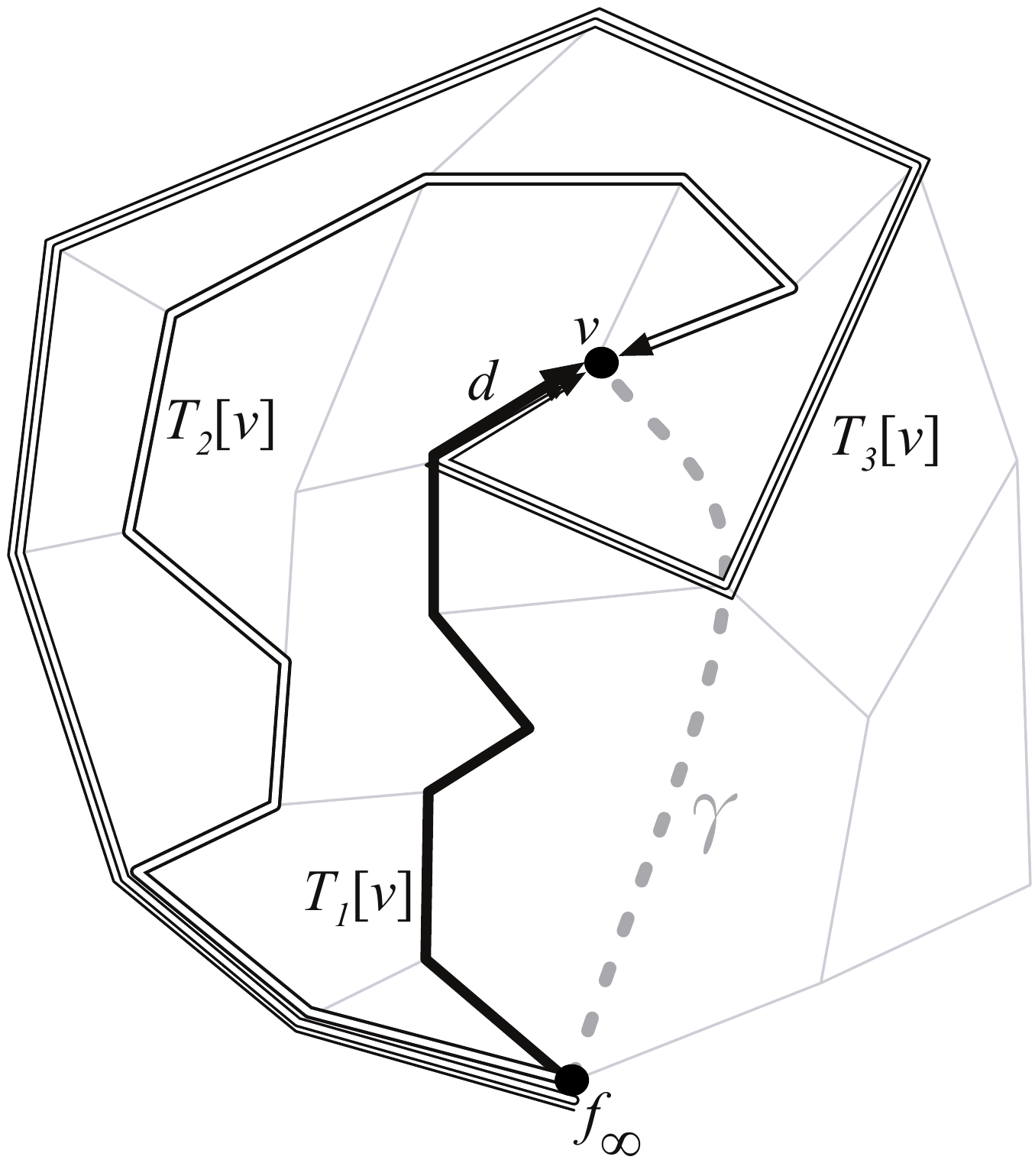}}
}{
\centerline{\includegraphics[scale=0.3]{figures/windings2BW.pdf}}
}
\caption{Bounding the number of pivots of dart $d$ whose head is a
  node $v$ of ${G^*}$. 
Three tree paths to $v$ at times $t_1 < t_2 < t_3$ are
  shown. $T_3[v]$ (triple-lined) is left of $T_2[v]$ (double-lined) and $T_2[v]$ is left of
  $T_1[v]$ (black) . The dart $d$ is in $T_1[v]$ and
  $T_3[v]$, but not in
  $T_2[v]$. Therefore, it must have pivoted out of $T$ and into $T$
  between time $t_1$ and $t_3$. 
 The curve $\gamma$ (dashed) visits
  only nodes of ${G^*}$. The winding number of $T_1[v]$ about $\gamma$ is
  0. The winding number of $T_3[v]$ must be greater than that of $T_1$ (1 in
  this example).} 
\label{fig:winding}
\end{figure}

Consider a dart $d$ of ${G^*}$ with head $v$. 
Let $\gamma$ be an arbitrary root$(T)$-to-$v$ curve on the sphere. 
Let $T_1[v]$ and $T_2[v]$ denote the tree path to $v$ at two distinct
times in the execution of the algorithm.
Since right-shortness is
preserved throughout the algorithm and since with every pivot
$\length(T[v])$ may only decrease, if $T_2[v]$ occurs later in the
execution than $T_1[v]$, then $T_2[v]$ is left of $T_1[v]$.
It follows that the winding number of $T[v]$ about $\gamma$ between any two occurrences of $d$  as
the pivot dart in line~\ref{relax} must increase. See Fig.~\ref{fig:winding}. 
Therefore, if we let $T_0$ denote the initial tree $T$, and $T_t$
denote the tree $T$ at the latest time node $v$ appears in ${G^*}$, then
the number of times $d$ may appear as the pivot dart in 
line~\ref{relax} is bounded by the difference of the winding numbers
of $T_t[v]$ and $T_0[v]$ about $\gamma$. Since $\gamma$ is arbitrary, we may choose
it to intersect the embedding of ${G^*}$ only at nodes, and to further
require that it visit the minimum possible number of nodes of
${G^*}$. This number is bounded by the diameter of the face-vertex
incidence graph of $G$, which is bounded by the minimum of the
diameter of $G$ and the diameter of ${G^*}$. Since both $T_0$ and $T$ are
simple, the absolute value of their winding number about $\gamma$ is trivially bounded
by the number of nodes $\gamma$ visits.  Therefore, the total number of pivots is bounded
by $\sum_d \diam = O(\diam \cdot n)$.
The total running time of the algorithm is thus bounded by $O(\diam \cdot
n \log n)$.

\bibliographystyle{plain}
\bibliography{diameter}

\end{document}